\documentclass[10pt,twocolumn]{article} 
\usepackage{latex8}
\usepackage{times}

\usepackage{verbatim} 
\usepackage{clrscode3e}
\usepackage[pdftex]{graphicx}
\usepackage{setspace}
\usepackage[cmex10]{amsmath}
\usepackage{amsmath, amsthm, amssymb}
\usepackage{algorithmic}
\usepackage{algorithm}
\usepackage{subfig}
\usepackage{caption}
\usepackage{paralist}
\usepackage{cases}
\usepackage{cite}

\usepackage{tikz}

\usepackage{etoolbox}

\usepackage{color}

\newcommand{\ceiling}[1]{\left\lceil{#1}\right\rceil}
\newcommand{\floor}[1]{\left\lfloor{#1}\right\rfloor}
\newcommand{\setof}[1]{\left\{{#1}\right\}}

\newcommand{\frameworkku}[1]{$\mathbf{k^2U}$}
\newcommand{\frameworkkq}[1]{$\mathbf{k^2Q}$}

 \def\myendproof{{\ \vbox{\hrule\hbox{%
   \vrule height1.3ex\hskip0.8ex\vrule}\hrule }}\par}
 \renewenvironment{proof}{\noindent{\bf Proof. }}{\myendproof}
 
 \setboolean{ALC@noend}{true}
\providebool{techreport}
\setbool{techreport}{true}

\newtheorem{theorem}{Theorem}
\newtheorem{lemma}{Lemma}
\newtheorem{corollary}{Corollary}

\newtheorem{definition}{Definition}

\graphicspath{{fig/multiframe/}{fig/arbitrary/}{fig/dag/}} 

\usetikzlibrary{%
  arrows,%
  shapes.misc,
  shapes.arrows,%
  chains,%
  matrix,%
  positioning,
  scopes,%
  decorations.pathmorphing,
  shadows%
}

\pgfdeclarelayer{background}
\pgfdeclarelayer{foreground}
\pgfsetlayers{background,main,foreground}
\tikzstyle{materia}=[draw, fill=white, text width=1.0em, text centered,
  minimum height=4em,drop shadow]
\tikzstyle{practica} = [materia, text width=18em, minimum width=8em,
align =left,
  minimum height=3em, rounded corners, drop shadow]
\tikzstyle{texto} = [above, text width=6em, text centered]
\tikzstyle{linepart} = [draw, thick, color=blue!50, -latex', dashed]
\tikzstyle{line} = [draw, line width = 2pt, color=blue!50, -latex']
\tikzstyle{ur}=[draw, text centered, minimum height=0.01em]

\title{Automatic Parameter Derivations in k2U Framework}

\author{
    Jian-Jia Chen and Wen-Hung Huang\\
    Department of Informatics\\
    TU Dortmund University, Germany
    \and
    Cong Liu\\
    Department of Computer Science\\
    The University of Texas at Dallas
}
\vspace{3mm}

\begin{document}

\maketitle

\begin{abstract}
  We have recently developed a general schedulability test framework,
  called \frameworkku{}, which can be applied to deal with a large
  variety of task models 
  that have been widely studied in
  real-time embedded systems.  The \frameworkku{} framework provides
  several means for the users to convert arbitrary schedulability
  tests (regardless of platforms and task models) into polynomial-time
  tests with closed mathematical expressions. However, the
  applicability (as well as the performance) of the \frameworkku{}
  framework relies on the users to index the tasks properly and define
  certain constant parameters.

  This report describes how to automatically index the tasks properly and derive those
  parameters. We will cover several typical schedulability tests in
  real-time systems to explain how to systematically and automatically
  derive those parameters required by the \frameworkku{}
  framework. This automation significantly empowers the \frameworkku{}
  framework to handle a wide range of classes of real-time execution
  platforms and task models, including uniprocessor scheduling,
  multiprocessor scheduling, self-suspending task systems, real-time
  tasks with arrival jitter, services and virtualizations with
  bounded delays, etc.  
\end{abstract}

\section{Introduction}

To analyze the worst-case response time or to ensure the timeliness of
the system, for each of individual task and platform models, researchers tend to develop
dedicated techniques that result in schedulability tests with
different time/space complexity and accuracy of the
analysis.
A very widely adopted case is the schedulability test of a
(constrained-deadline) sporadic real-time task $\tau_k$ under
fixed-priority scheduling in uniprocessor systems, in which the
time-demand analysis (TDA) developed in
\cite{DBLP:conf/rtss/LehoczkySD89} can be adopted. That is, if
\begin{equation}
  \label{eq:exact-test-constrained-deadline}
\exists t \mbox{ with } 0 < t \leq D_k {\;\; and \;\;} C_k +
\sum_{\tau_i \in hp(\tau_k)} \ceiling{\frac{t}{T_i}}C_i \leq t,
\end{equation}
then task $\tau_k$ is schedulable under the fixed-priority scheduling algorithm, where $hp(\tau_k)$ is the set of tasks with higher priority than $\tau_k$, $D_i$, $C_i$, and $T_i$ represent $\tau_i$'s relative deadline, worst-case execution time, and period, respectively. TDA requires pseudo-polynomial-time complexity to check the time points that lie in $(0, D_k]$ for Eq.~\eqref{eq:exact-test-constrained-deadline}.  The utilization $U_i$ of a sporadic task $\tau_i$ is defined as $C_i/T_i$.

However, it is not always necessary to test all possible time points
to derive a safe worst-case response time or to provide sufficient
schedulability tests.
The general and key concept to obtain sufficient schedulability tests in
\frameworkku{} in \cite{DBLP:journals/corr/abs-1501.07084} and
\frameworkkq{} in \cite{DBLP:journals/corr/abs-k2q}  is to test only a
subset of such points for verifying the schedulability. 
Traditional fixed-priority schedulability tests often have
pseudo-polynomial-time (or even higher) complexity. 
The idea
implemented in the \frameworkku{} and \frameworkkq{} frameworks  is to 
provide a general $k$-point schedulability test, which
only needs to test $k$ points under \textit{any} fixed-priority
scheduling when checking schedulability of the task with the $k^{th}$
highest priority in the system.  
Suppose that there are $k-1$ higher priority tasks, indexed as $\tau_1, \tau_2, \ldots, \tau_{k-1}$, than task $\tau_k$.
The success of the \frameworkku{} framework is based on a 
$k$-point effective schedulability test, defined as follows:
\begin{definition}[Chen et al. \cite{DBLP:journals/corr/abs-1501.07084}]
  \label{def:kpoints-k2u}
  A $k$-point effective schedulability test is a sufficient schedulability test of a fixed-priority scheduling policy, that verifies the existence of $t_j \in \setof{t_1, t_2, \ldots t_k}$ with $0 < t_1 \leq t_2 \leq \cdots \leq t_k$ such that \begin{equation}
    \label{eq:precodition-schedulability-k2u}
    C_k + \sum_{i=1}^{k-1} \alpha_i t_i U_i + \sum_{i=1}^{j-1} \beta_i t_i U_i \leq t_j,
  \end{equation}
  where $C_k > 0$, $\alpha_i > 0$, $U_i > 0$, and $\beta_i >0$ are dependent upon the setting
  of the task models and task $\tau_i$. \myendproof
\end{definition}

The \frameworkku{} framework \cite{DBLP:journals/corr/abs-1501.07084}
assumes that the corresponding coefficients $\alpha_i$ and $\beta_i$
in Definition~\ref{def:kpoints-k2u} are given. How to derive them
depends on the task models, the platform models, and the scheduling
policies.  Provided that these coefficients $\alpha_i$, $\beta_i$,
$C_i$, $U_i$ for every higher priority task $\tau_i$ are given, the
\frameworkku{} framework can find the worst-case assignments of the
values $t_i$ for the higher-priority tasks $\tau_i$.

Although several applications were adopted to demonstrate the power
and the coverage of the \frameworkku{} framework, we were not able to
provide an automatic procedure to construct the required coefficients
$\alpha_i$ and $\beta_i$ in Definition~\ref{def:kpoints-k2u} in
\cite{DBLP:journals/corr/abs-1501.07084}. Instead, we stated in
\cite{DBLP:journals/corr/abs-1501.07084} as follows:
\begin{quote}\emph{
  The choice of good parameters $\alpha_i$ and $\beta_i$ affects the
  quality of the resulting schedulability bounds. .....  However, deriving
  the \emph{good} settings of $\alpha_i$ and $\beta_i$ is actually not
  the focus of this paper. The framework does not care how the
  parameters $\alpha_i$ and $\beta_i$ are obtained. The framework
  simply derives the bounds according to the given parameters
  $\alpha_i$ and $\beta_i$, regardless of the settings of $\alpha_i$
  and $\beta_i$. The correctness of the settings of $\alpha_i$ and
  $\beta_i$ is not verified by the framework.}
\end{quote}

{\bf Contributions:} This report explains how to automatically derive
those parameters needed in the \frameworkku{} framework.  We will
cover several typical schedulability tests in real-time systems to
explain how to systematically and automatically derive those
parameters required by the \frameworkku{} framework. This automation
significantly empowers the \frameworkku{} framework to handle a wide
range of classes of real-time execution platforms and task models,
including uniprocessor scheduling, multiprocessor scheduling,
self-suspending task systems, real-time tasks with arrival jitter,
services and virtualizations with bounded delays, etc.  More
precisely, if the corresponding (exponential time or
pseudo-polynomial-time) schedulability test is in one of the classes
provided in this report, the derivations of the hyperbolic-form
schedulability tests, utilization-based analysis, etc. can be
automatically constructed.

Given an arbitrary schedulability test, there are many
ways to define a corresponding k-point effective schedulability test.
The constructions of the coefficients in this report may not be the
best choices. All the constructions in this
report follow the same design philosophy: \emph{We first identify the
  tasks that can release at least one more job at time $0 < t < D_k$
  in the schedulability test and define the effective test point of
  such a task at its last release before $D_k$.} There may be other
more effective constructions for different schedulability tests. These
opportunities are not explored in this report.

\noindent\textbf{Organizations.} 
The rest of this report is organized
as follows:
\begin{compactitem}
\item The basic terminologies and models are presented in
  Section~\ref{sec:model}. 
\item We will present three classes of applicable schedulability tests,
which can allow automatic parameter derivations:
\begin{compactitem}
\item {\bf Constant inflation} in Section~\ref{sec:constant-inflation}: This class covers a wide range of
  applications in which the workload of a higher-priority task may
  have a constant inflation to quantify the additional workload in the
  analysis window.
\item {\bf Bounded-delayed service} in Section~\ref{sec:different-service}: This class covers a wide range of
  applications in which the computation service provided to the task
  system can be lower bounded by a constant slope with a constant
  offset.
\item {\bf Arrival jitter} in Section~\ref{sec:jitter}: This class covers a wide range of
  applications in which a higher-priority task may have arrival jitter
  in the analysis window.
\end{compactitem}
\end{compactitem}
Please note that we will not specifically explain how to use the
\frameworkku{} framework in this report. Please refer
to \cite{DBLP:journals/corr/abs-1501.07084} for details. However, for completeness
the key lemmas in \cite{DBLP:journals/corr/abs-1501.07084} will be
summarized in Section~\ref{sec:model}.

\section{Models and Terminologies}
\label{sec:model}

\subsection{Basic Task and Scheduling Models}

This report will introduce the simplest settings by using the ordinary
sporadic real-time task model, even though the frameworks target at
more general task models. 
We define the terminologies here for
completeness.  A sporadic task $\tau_i$ is released repeatedly, with
each such invocation called a job. The $j^{th}$ job of $\tau_i$,
denoted $\tau_{i,j}$, is released at time $r_{i,j}$ and has an
absolute deadline at time $d_{i,j}$. Each job of any task $\tau_i$ is
assumed to have $C_i$ as its worst-case execution time.  The response
time of a job is defined as its finishing time minus its release
time. Associated with each task $\tau_i$ are a period $T_i$, which
specifies the minimum time between two consecutive job releases of
$\tau_i$, and a deadline $D_i$, which specifies the relative deadline
of each such job, i.e., $d_{i,j}=r_{i,j}+D_i$. The worst-case response
time of a task $\tau_i$ is the maximum response time among all its
jobs.  The utilization of a task $\tau_i$ is defined as $U_i=C_i/T_i$.

A sporadic task system $\tau$ is said to be an implicit-deadline
task system if $D_i = T_i$ holds for each $\tau_i$. A sporadic task system
$\tau$ is said to be a constrained-deadline task system if $D_i \leq T_i$
holds for each $\tau_i$.  Otherwise, such a sporadic task system
$\tau$ is an arbitrary-deadline task system.

A task is said \emph{schedulable} by a scheduling policy if all of its
jobs can finish before their absolute deadlines, i.e., the worst-case
response time of the task is no more than its relative deadline.  A
task system is said \emph{schedulable} by a scheduling policy if all
the tasks in the task system are schedulable. A \emph{schedulability
  test} is to provide sufficient conditions to ensure the feasibility
of the resulting schedule by a scheduling policy. 

Throughout the report, we will focus on fixed-priority
scheduling. That is, each task is associated with a priority level.
We will only present the schedulability test of a certain task
$\tau_k$, that is under analysis. For notational brevity, in the
framework presentation, we will implicitly assume that there are $k-1$
tasks, says $\tau_1, \tau_2, \ldots, \tau_{k-1}$ with higher-priority
than task $\tau_k$. \emph{These $k-1$ higher-priority tasks are
  assumed to be schedulable before we test task $\tau_k$.} We will use
$hp(\tau_k)$ to denote the set of these $k-1$ higher priority tasks,
when their orderings do not matter. Moreover, we only consider the
cases when $k \geq 2$, since $k=1$ is usually trivial.

Note that different task models may have different terminologies
regarding to $C_i$ and $U_i$. Here, we implicitly assume that $U_i$ is
always $C_i/T_i$. The definition of $C_i$ can be very dependent upon
the task systems.

\subsection{Properties of \frameworkku{}}

By using the property defined in Definition~\ref{def:kpoints-k2u}, we
can have the following lemmas in the \frameworkku{} framework
\cite{DBLP:journals/corr/abs-1501.07084}. All the proofs of the
following lemmas are in
\cite{DBLP:journals/corr/abs-1501.07084}.

\begin{lemma}[Chen et al. \cite{DBLP:journals/corr/abs-1501.07084}]
\label{lemma:framework-constrained-k2u}
For a given $k$-point effective schedulability test of a scheduling 
algorithm, defined in
Definition~\ref{def:kpoints-k2u},
in which $0 < t_k$ and $0 < \alpha_i \leq \alpha$, and $0 < \beta_i \leq \beta$ for any
$i=1,2,\ldots,k-1$, task $\tau_k$ is schedulable by the scheduling
algorithm if the following condition holds 
\begin{equation}
\label{eq:schedulability-constrained-k2u}
\frac{C_k}{t_k} \leq \frac{\frac{\alpha}{\beta}+1}{\prod_{j=1}^{k-1} (\beta U_j + 1)} - \frac{\alpha}{\beta}.
\end{equation}
\end{lemma}

\begin{lemma}[Chen et al. \cite{DBLP:journals/corr/abs-1501.07084}]
\label{lemma:framework-totalU-constrained-k2u}
For a given $k$-point effective schedulability test of a scheduling
algorithm, defined in
Definition~\ref{def:kpoints-k2u},
in which $0 < t_k$ and $0 < \alpha_i \leq \alpha$ and $0 < \beta_i \leq \beta$ for any
$i=1,2,\ldots,k-1$, task $\tau_k$ is schedulable by the scheduling
algorithm if 
\begin{equation}
\label{eq:schedulability-totalU-constrained-k2u}
\frac{C_k}{t_k} + \sum_{i=1}^{k-1}U_i \leq \frac{(k-1)((\alpha+\beta)^{\frac{1}{k}}-1)+((\alpha+\beta)^{\frac{1}{k}}-\alpha)}{\beta}.
\end{equation}
\end{lemma}

\begin{lemma}[Chen et al. \cite{DBLP:journals/corr/abs-1501.07084}]
\label{lemma:framework-totalU-exclusive-k2u}
For a given $k$-point effective schedulability test of a scheduling
algorithm, defined in
Definition~\ref{def:kpoints-k2u},
in which $0 < t_k$ and $0 < \alpha_i \leq \alpha$ and $0 < \beta_i \leq \beta$ for any
$i=1,2,\ldots,k-1$, task $\tau_k$ is schedulable by the scheduling
algorithm if 
\begin{equation}
\label{eq:schedulability-totalU-exclusive-k2u}
\beta \sum_{i=1}^{k-1}U_i \leq \ln(\frac{\frac{\alpha}{\beta}+1}{\frac{C_k}{t_k}+\frac{\alpha}{\beta}}).
\end{equation}
\end{lemma}

\begin{lemma}[Chen et al. \cite{DBLP:journals/corr/abs-1501.07084}]
\label{lemma:framework-general-k2u}
For a given $k$-point effective schedulability test of a fixed-priority scheduling
algorithm, defined in
Definition~\ref{def:kpoints-k2u}, task $\tau_k$ is schedulable by the scheduling algorithm,
in which $0 < t_k$ and $0 < \alpha_i$ and $0 < \beta_i$ for any
$i=1,2,\ldots,k-1$, 
 if
the following condition holds 
\begin{equation}
\label{eq:schedulability-general-k2u}
0 < \frac{C_k}{t_k} \leq 1 -  \sum_{i=1}^{k-1}  \frac{U_i(\alpha_i
  +\beta_i)}{\prod_{j=i}^{k-1} (\beta_jU_j + 1)}.
\end{equation}
\end{lemma}

\section{Classes of Applicable Schedulability Tests}
\label{sec:different-arrival}

We will present three classes of applicable schedulability tests,
which can allow automatic parameter derivations:
\begin{itemize}
\item {\bf Constant inflation}: This class covers a wide range of
  applications in which the workload of a higher-priority task may
  have a constant inflation to quantify the additional workload in the
  analysis window.
\item {\bf Bounded delayed service}: This class covers a wide range of
  applications in which the computation service provided to the task
  system can be lower bounded by a constant slope with a constant
  offset.
\item {\bf Arrival jitter}: This class covers a wide range of
  applications in which a higher-priority task may have arrival jitter
  in the analysis window.
\end{itemize}
\subsection{Constant Inflation}
\label{sec:constant-inflation}

Suppose that the schedulability test is as follows:
\begin{equation}
  \label{eq:test-constant-inflation}
  \exists 0 < t \leq D_k \mbox{ s.t. } C_k + \sum_{\tau_i \in hp(\tau_k)}\sigma \left(\ceiling{\frac{t}{T_i}} C_i + bC_i\right) \leq t,
\end{equation}
where $\sigma > 0$ and $b \geq 0$.
We now classify the task set
$hp(\tau_k)$ into two subsets:
\begin{itemize}
\item $hp_1(\tau_k)$ consists of the higher-priority tasks with periods
  smaller than $D_k$.
\item $hp_2(\tau_k)$ consists of the higher-priority tasks with periods
  greater than or equal to $D_k$.
\end{itemize}
Therefore, we can rewrite Eq.~\eqref{eq:test-constant-inflation} to 
\begin{equation}
  \label{eq:test-constant-inflation2}
  \exists 0 < t \leq D_k \mbox{ s.t. } C_k' + \sum_{\tau_i \in hp_1(\tau_k)}\sigma \left(\ceiling{\frac{t}{T_i}} C_i + bC_i\right) \leq t,
\end{equation}
where $C_k'$ is defined as $C_k + \sum_{\tau_i \in hp_2(\tau_k)}\sigma (1+b)C_i$.

\begin{theorem}
  \label{thm:constant-inflation}
  For 
  Eq.~\eqref{eq:test-constant-inflation2}, the $k$-point effective
  schedulability test in Definition~\ref{def:kpoints-k2u} is with the
  following settings:
  \begin{compactitem}
  \item $t_k = D_k$,
  \item  for $\tau_i \in hp_1(\tau_k)$, $t_i =\left (\ceiling{\frac{D_k}{T_i}}-1\right)T_i = g_i T_i$,
  \item  for $\tau_i \in hp_1(\tau_k)$,  the parameter $\alpha_i$ is $\frac{\sigma(g_i+b)}{g_i}$ with $0 < \alpha_i \leq \sigma(1+b)$, and 
  \item  for $\tau_i \in hp_1(\tau_k)$,  the parameter $\beta_i$ is $\frac{\sigma}{g_i}$ with $0 < \beta_i \leq \sigma$.
  \end{compactitem}
  The tasks in $hp_1(\tau_k)$ are indexed according to non-decreasing
  $t_i$ defined above to satisfy Definition~\ref{def:kpoints-k2u}.
\end{theorem}
\begin{proof}
  Let $t_i$ be $\left (\ceiling{\frac{D_k}{T_i}}-1\right)T_i = g_i
  T_i$, where $g_i$ is an integer. By the definition of
  $hp_1(\tau_k)$, we know that $g_i \geq 1$. We index the tasks in
  $hp_1(\tau_k)$ according to non-decreasing $t_i$. We assume that
  there are $k-1$ tasks in $hp(\tau_k)$ for notational
  brevity. 

  Therefore, the left-hand side of Eq.~\eqref{eq:test-constant-inflation2} at time
  $t=t_j$ upper bounded by {\small \begin{align}
      &   C_k' + \sum_{i=1}^{k-1}\sigma \left(\ceiling{\frac{t_j}{T_i}} C_i + bC_i\right)\nonumber\\
      \leq \;\;& C_k' + \sum_{i=1}^{j-1}\sigma \left(\ceiling{\frac{D_k}{T_i}} C_i + bC_i\right) + \sum_{i=j}^{k-1}\sigma \left(\ceiling{\frac{t_i}{T_i}} C_i + bC_i\right)\nonumber\\
      = \;\;& C_k' + \sum_{i=1}^{j-1}\sigma \left((g_i+1) C_i + bC_i\right) + \sum_{i=j}^{k-1}\sigma \left(g_i C_i + bC_i\right)\nonumber\\
      = \;\;& C_k' + \sum_{i=1}^{k-1}\sigma \left(g_i C_i + bC_i\right) + \sum_{i=1}^{j-1} \sigma\cdot C_i \nonumber\\
      =_1\;\; & C_k' + \sum_{i=1}^{k-1} \frac{\sigma(g_i+b)}{g_i} t_i
      U_i + \sum_{i=1}^{j-1} \frac{\sigma}{g_i}t_i U_i,
    \end{align}}where the inequality comes from $t_1 \leq t_2 \leq
  \cdots \leq t_k = D_k$ in our index rule, and $=_1$ comes from the
  setting that $C_i = U_i T_i = \frac{1}{g_i} t_i U_i$. That is, the
  test in Eq. \eqref{eq:test-constant-inflation2} can be safely
  rewritten as
  \[
(\exists t_j | j=1,2,\ldots, k),\qquad C_k' + \sum_{i=1}^{k-1} \alpha_i t_i
      U_i + \sum_{i=1}^{j-1} \beta_i t_i U_i \leq t.
  \]

Therefore,
  we can conclude the compatibility of the test with the
  \frameworkku{} framework by setting $\alpha_i=\frac{\sigma(g_i+b)}{g_i}$ and $\beta_i=\frac{\sigma}{g_i}$.
  Due to the fact that $g_i \geq 1$, we also know that $0 < \alpha_i = \sigma(1+\frac{b}{g_i}) \leq \sigma(1+b)$ and $0 < \beta_i \leq \sigma$. This concludes the proof. 
\end{proof}

We can now directly apply
Lemmas~\ref{lemma:framework-constrained-k2u},
\ref{lemma:framework-totalU-constrained-k2u}, and
\ref{lemma:framework-totalU-exclusive-k2u} for the test in
Eq.~\eqref{eq:test-constant-inflation2}.
\begin{corollary}
\label{cor:inflation}
  For a schedulability test in
  Eq.~\eqref{eq:test-constant-inflation2}, task $\tau_k$ is
  schedulable if
  \begin{equation}
    \label{eq:hyperbolic-form1}
\left(\frac{C_k'}{D_k} + (1+b)\right)\prod_{\tau_i \in hp_1(\tau_k)} (\sigma U_i+1) \leq 2+b,
  \end{equation}
  or if
  \begin{equation}
    \label{eq:hyperbolic-form2}
\sigma\sum_{\tau_i \in hp_1(\tau_k)}U_i \leq \ln\left(\frac{2+b}{\frac{C_k'}{D_k} + 1+b}\right)
  \end{equation}
\end{corollary}
\begin{proof}
  This comes directly from Theorem~\ref{thm:constant-inflation} and
  Lemma~\ref{lemma:framework-constrained-k2u} and
  Lemma~\ref{lemma:framework-totalU-exclusive-k2u}.
\end{proof}

\subsubsection{Applications}
 This class of schedulability tests in
Eq.~\eqref{eq:test-constant-inflation} covers quite a lot of cases in
both uniprocessor and multiprocessor systems. 

\underline{In uniprocessor systems:}
\begin{itemize}
\item Constrained-deadline and implicit-deadline uniprocessor task scheduling \cite{liu73scheduling,journals/pe/LeungW82}: A simple schedulability test for this case is to set $\sigma=1$ and $b=0$ in Eq.~\eqref{eq:test-constant-inflation}. This is used to demonstrate the usefulness of the \frameworkku{} framework in \cite{DBLP:journals/corr/abs-1501.07084}.

\item Uniprocessor non-preemptive scheduling \cite{DBLP:conf/ecrts/BruggenCH15}: This is a known case
  in which $C_k$ should be set to $C_k + \max_{\tau_i \in lp(\tau_k)}
  C_{i}$, $\sigma=1$ and $b=0$ in Eq.~\eqref{eq:test-constant-inflation}, where $lp(\tau_k)$ is the set of the
  lower-priority tasks than task $\tau_k$. This is implicitly used in \cite{DBLP:conf/ecrts/BruggenCH15}.

\item Bursty-interference \cite{RTSS14a}: This is a known case in
  which $\sigma=1$ and $b$ is set to a constant to reflect the bursty
  interference for the first job in the analysis window in Eq.~\eqref{eq:test-constant-inflation}. It is shown
  in \cite{RTSS14a} that this can be used to model the schedulability
  analysis of deferrable servers and self-suspending task systems (by
  setting $C_k$ to $C_k + S_k$, where $S_k$ is the maximum
  self-suspending time of task $\tau_k$).

\end{itemize}

\underline{In multiprocessor systems} with $M$ processors and constrained deadline
task sets:
\begin{itemize}
\item Multiprocessor global DM/RM scheduling for sporadic task
  systems: A simple schedulability test in this case is to set $\sigma=\frac{1}{M}$ and $b=1$ in Eq.~\eqref{eq:test-constant-inflation}. This
  is used to demonstrate the usefulness of the \frameworkku{}
  framework in \cite{DBLP:journals/corr/abs-1501.07084}.
\item Multiprocessor global DM/RM scheduling for self-suspending task
  systems and directed-acyclic-graph (DAG) task structures: This is
  similar to the above case for sporadic task systems in which
  $\sigma=\frac{1}{M}$ and $b=1$ by setting different equivalent values of $C_k$ in Eq.~\eqref{eq:test-constant-inflation}. For details, please refer to
  \cite{DBLP:journals/corr/abs-1501.07084}.
\item Multiprocessor partitioned RM/DM scheduling for sporadic task
  systems: Testing whether a task $\tau_k$ can be feasibly assigned
  \emph{statically} on a processor can be done by setting
  $\sigma=\frac{1}{M}$ and $b=0$ in
  Eq.~\eqref{eq:test-constant-inflation}. This is used in
  \cite{DBLP:journals/corr/Chen15k} for improving the speedup factors
  and utilization-based schedulability tests.
\end{itemize}

\subsection{Bounded Delay Services}
\label{sec:different-service}

We now discuss another class of schedulability tests by considering
\emph{bounded services}. In the class of the schedulability tests in
Eq.~\eqref{eq:test-constant-inflation}, the right-hand side of the
inequality is always $t$. Here, in this subsection, we will change the
right-hand side of Eq.~\eqref{eq:test-constant-inflation} to $A(t)$,
where $A(t)$ is defined to quantify the minimum service provided by
the system in any interval length $t > 0$ (after the normalization for the
schedulability test of task $\tau_k$).  We will consider the following
schedulability test for verifying the schedulability of task $\tau_k$:
\begin{equation}
  \label{eq:test-bounded-service}
  \exists 0 < t \leq D_k \mbox{ s.t. } C_k + \sum_{\tau_i \in hp(\tau_k)}\sigma \left(\ceiling{\frac{t}{T_i}} C_i + bC_i\right) \leq A(t),
\end{equation}
where $\sigma > 0$ and $b \geq 0$ are constants. We will specifically
consider two types of $A(t)$:
\begin{itemize}
\item {\it Segmented service curves}: An example of such a case is the
  time division multiple access (TDMA) arbitrary policy
  \cite{DBLP:conf/rtcsa/Sha03,WTVL06} to provide fixed time slots with
  $\sigma C_{slot}$ total amount of service in every TDMA cycle length
  $T_{cycle}$.  In this case, we consider that 
  \begin{equation}
    \label{eq:service-curve-TDMA}
    A(t) = t - \ceiling{\frac{t}{T_{cycle}}}\cdot (T_{cycle} - \sigma C_{slot}).
  \end{equation}
  where $T_{cycle}$ and $C_{slot}$ are specified as constants.  Note
  that the setting of $A(t)$ in Eq.~\eqref{eq:service-curve-TDMA} is
  an approximation of the original TDMA service curve, to be discussed
  later.

\item {\it Bounded delay service curves}: The service provided by the
  system is lower bounded by a constant slope $\gamma$ when $t \geq
  t_{delay}$, where $\gamma$ and $t_{delay}$ are specified as
  constants. Specifically, in this case,
  \begin{equation}
    \label{eq:service-curve-bounded-delay}
    A(t) = \max\{0, \gamma(t-t_{delay})\}.
  \end{equation}
\end{itemize}

Figure~\ref{fig:bounded-functions} provides an example for the above
two cases. We will discuss how these two bounds in
Eq.~\eqref{eq:service-curve-TDMA} and
Eq.~\eqref{eq:service-curve-bounded-delay} are related to TMDA and
other hierarchical scheduling policies.

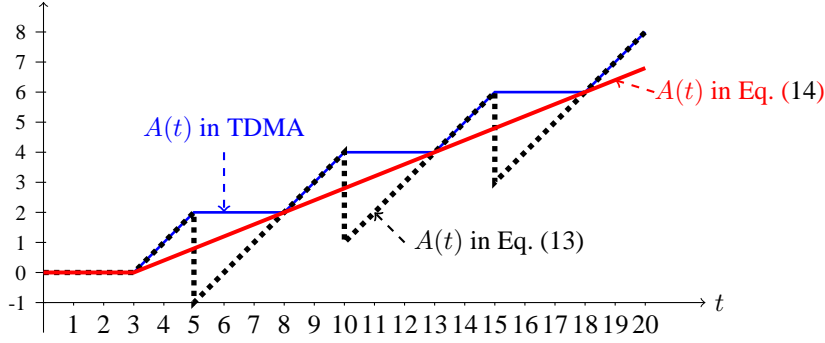
\begin{figure*}[t]
\centering
  \begin{tikzpicture}[node distance=0cm, y=0.4cm, x=0.4cm]
    \draw[->] (0,-1) -- coordinate (x axis mid) (22,-1) node[right]{$t$};
    \draw[->] (0,-2) -- coordinate (y axis mid) (0,9) node[above]{};
    \foreach \x in {1,...,20}
    \draw (\x,-1) -- (\x,-1.1) 
    node[anchor=north] { \x}; 
    \foreach \y in {-1,0,1,...,8}
    \draw (1pt, \y) -- (-3pt, \y) 
    node[anchor=east] {\footnotesize \y}; 
    
      \draw[line width = 1pt,color=blue] plot[] 
      (0, 0) -- (3, 0) -- (5, 2) -- (8, 2) -- (10, 4) -- (13, 4) -- (15, 6) -- (18, 6) -- (20, 8);
      \draw[line width = 2pt,color=black, dotted] plot[] 
      (0, 0) -- (3, 0) -- (5, 2) -- (5, -1) -- (8, 2) -- (10, 4) -- (10, 1) -- (13, 4) -- (15, 6) -- (15, 3) -- (18, 6) -- (20, 8);
      \draw[line width = 1.5pt,color=red] plot[] 
      (0, 0) -- (3, 0) -- (18, 6) -- (20, 6.8);
      \draw[black,thick,dashed,->](12, 1) -- (11, 2);
      \draw[black](12,1) node[anchor=west]{$A(t)$ in Eq.~\eqref{eq:service-curve-TDMA}};
      \draw[red,thick,dashed,->](20.3, 6) -- (19, 6.4);
      \draw[red](20,6) node[anchor=west]{$A(t)$ in Eq.~\eqref{eq:service-curve-bounded-delay}};
      \draw[blue,thick,dashed,->](6, 4) -- (6, 2);
      \draw[blue](6,4) node[anchor=south]{$A(t)$ in TDMA};
      
    \end{tikzpicture}
    \caption{An example of delayed service curve: $\sigma=1$,
      $C_{slot}=2$, and $T_{cycle}=5$, where $\gamma=\frac{\sigma
        C_{slot}}{T_{cycle}} \leq 1$ and $t_{delay}=T_{cycle} - \sigma
      C_{slot}$ in Eq.~\eqref{eq:service-curve-bounded-delay}.}
  \label{fig:bounded-functions}
\end{figure*}

\subsubsection{Segmented service curve: $A(t)$ in  Eq.~\eqref{eq:service-curve-TDMA}}
\label{sec:segmented} 
For Eq.~\eqref{eq:test-bounded-service}, in which $A(t)$ is defined in
Eq.~\eqref{eq:service-curve-TDMA}, the schedulability test of task $\tau_k$ is as
follows:
{\small 
\begin{align}
 & \qquad\qquad\exists 0 < t \leq D_k \mbox{ s.t. }\nonumber\\
& C_k + \sum_{\tau_i \in hp(\tau_k)}\sigma \left(\ceiling{\frac{t}{T_i}} C_i + bC_i\right)\nonumber\\
\leq\;\;& t - \ceiling{\frac{t}{T_{cycle}}}\cdot (T_{cycle} - \sigma C_{slot})   \label{eq:test-bounded-TDMA}
\end{align}}
where $\sigma > 0$, $b \geq 0$, $C_{slot}$, and $T_{slot}$ are constants with $T_{cycle} - \sigma C_{slot} \geq 0$. The above test can be reorganized as 
{\small\begin{align}
&  \exists 0 < t \leq D_k \mbox{ s.t. }\nonumber\\
& C_k + \sigma\left(\ceiling{\frac{t}{T_{cycle}}} (\frac{T_{cycle}}{\sigma} - C_{slot}) \right) \nonumber\\
&+ \sum_{\tau_i \in hp(\tau_k)}\sigma \left(\ceiling{\frac{t}{T_i}} C_i + bC_i\right) \qquad\leq t.   \label{eq:test-bounded-TDMA-final-v1}
\end{align}}
The above test can be imagined as if there is a virtual
higher-priority task $\tau_{virtual}$ with period $T_{cycle}$ and execution
time $\frac{T_{cycle}}{\sigma}-C_{slot}$. In this formulation, the
virtual task $\tau_{virtual}$ does not have any inflation.  If $C_k - \sigma
\cdot b \cdot(\frac{T_{cycle}}{\sigma} - C_{slot}) > 0$, we can
further set $C_k'$ as $C_k - \sigma \cdot b
\cdot(\frac{T_{cycle}}{\sigma} - C_{slot})$, and the schedulability
test of task $\tau_k$ becomes
{\small\begin{align}
&  \exists 0 < t \leq D_k \mbox{ s.t. } \nonumber\\
&C_k' + \sigma\left(\ceiling{\frac{t}{T_{cycle}}}\cdot (\frac{T_{cycle}}{\sigma} - C_{slot}) + b (\frac{T_{cycle}}{\sigma} - C_{slot})\right)\nonumber\\
&\qquad+ \sum_{\tau_i \in hp(\tau_k)}\sigma \left(\ceiling{\frac{t}{T_i}} C_i + bC_i\right) \leq t.   \label{eq:test-bounded-TDMA-final-v2}
\end{align}}

Therefore, we have reformulated the test to the same case in
Eq.~\eqref{eq:test-constant-inflation} by adding a virtual
higher-priority task $\tau_{virtual}$. We can directly use
Theorem~\ref{thm:constant-inflation} for this class of schedulability
tests.

\subsubsection{Bounded delay service curve: $A(t)$ in  Eq.~\eqref{eq:service-curve-bounded-delay}}
\label{sec:bounded-service-detail}

For Eq.~\eqref{eq:test-bounded-service}, in which $A(t)$ is defined in
Eq.~\eqref{eq:service-curve-bounded-delay}, the schedulability test of task $\tau_k$ is as
follows:
{\small \begin{align}
&  \exists t_{delay} < t \leq D_k \mbox{ s.t. }\nonumber\\
& C_k + \sum_{\tau_i \in hp(\tau_k)}\sigma
 \left(\ceiling{\frac{t}{T_i}} C_i + bC_i\right) \leq \gamma
 (t-t_{delay}),   \label{eq:test-bounded-delay}
\end{align}}
where $\sigma > 0$, $b \geq 0$, $\gamma > 0$, and $0 < t_{delay} < D_k$ are
constants.  This can be rewritten as 
\begin{align}\small
&   \exists t_{delay} < t \leq D_k \mbox{ s.t. } \nonumber\\
& \frac{C_k+\gamma t_{delay}}{\gamma} + \sum_{\tau_i \in hp(\tau_k)}\frac{\sigma}{\gamma} \left(\ceiling{\frac{t}{T_i}} C_i + bC_i\right) \leq t.   \label{eq:test-bounded-delay2}
\end{align} 
It is also clear that for any $0 < t \leq t_{delay}$, the above inequality never holds when $C_k > 0$. Therefore, we can change the boundary condition from $t_{delay}< t$ to $0 < t$ safely. That is, we have
{\small \begin{align}
&  \exists 0 < t \leq D_k \mbox{ s.t. } \nonumber\\
& \frac{C_k+\gamma t_{delay}}{\gamma} + \sum_{\tau_i \in
  hp(\tau_k)}\frac{\sigma}{\gamma} \left(\ceiling{\frac{t}{T_i}} C_i +
  bC_i\right) \leq t.   \label{eq:test-bounded-delay3}
\end{align} }
With the above reformulation, the test is similar to that in Eq.~\eqref{eq:test-constant-inflation},
where $\sigma$ in Eq.~\eqref{eq:test-constant-inflation} is defined as $\frac{\sigma}{\gamma}$, and $C_k$ in Eq.~\eqref{eq:test-constant-inflation}  is defined as $\frac{C_k+\gamma t_{delay}}{\gamma}$.
Therefore, this case is now reduced to the same case in
Eq.~\eqref{eq:test-constant-inflation}. We can directly
use Theorem~\ref{thm:constant-inflation} for this class of
schedulability tests.

\subsubsection{Applications for TDMA}

Suppose that the system provides a time division multiple access
(TDMA) policy to serve an implicit-deadline sporadic task system with
a TDMA cycle $T_{cycle}$ and a slot length $C_{slot}$. The bandwidth
of the TDMA is $\gamma=\frac{C_{slot}}{T_{cycle}}$. As shown in
\cite{DBLP:conf/aspdac/WandelerT06}, the service provided by the TDMA
policy in an interval length $t$ is at least
$\max\{\floor{\frac{t}{T_{cycle}}}C_{slot}, t -
\ceiling{\frac{t}{T_{cycle}}}\cdot(T_{cycle}-C_{slot})\}$. The service
curve can still be lower-bounded by ignoring the term
$\floor{\frac{t}{T_{cycle}}}C_{slot}$, which leads to $t
-\ceiling{\frac{t}{T_{cycle}}}\cdot (T_{cycle} - C_{slot})$, as a
segmented service curve described in Eq.~\eqref{eq:service-curve-TDMA}.
Another way is to use a linear approximation \cite{WTVL06}, as a
bounded delay service curve in
Eq.~\eqref{eq:service-curve-bounded-delay}, to quantify the lower
bound on the service provided by the TDMA. It can be imagined that the
service starts when $t_{delay}=T_{cycle}-C_{slot}$ with utilization
$\gamma=T_{cycle}/C_{slot}$. Therefore, the service provided by the
TDMA in an interval length $t$ is lower bounded by $\max\{0,
t-t_{delay}+ \gamma\cdot(t-t_{delay}\}$.

These two different approximations and the original TDMA service curve
are all presented in Figure~\ref{fig:bounded-functions}. 
By adopting the segmented service curve,  the schedulability test for task $\tau_k$ can
be described by Eq.~\eqref{eq:test-bounded-TDMA} with $\sigma=1$ and
$b=0$. By the result in Sec.~\ref{sec:segmented}, we can
directly conclude that $0 < \alpha_i \leq 1$ and $0 < \beta_i \leq 1$
for $\tau_i \in hp(\tau_k)$ under RM scheduling, and, hence, the
schedulability test of task $\tau_k$ if $T_{cycle} < T_k$ is 
\begin{align}
&\left(\frac{T_{cycle} - C_{slot}}{T_{cycle}}+1\right)(U_k+1)\prod_{\tau_i \in
  hp(\tau_k)}(U_i+1) \leq 2 \nonumber\\
\Rightarrow &\prod_{i=1}^{k} (U_i+1) \leq \frac{2}{2-\gamma}.   \label{eq:tdma-uni-rm}
\end{align}
Therefore, if $T_{cycle} < T_k$, we can conclude that the
utilization bound is $\sum_{i=1}^{k} U_i \leq
k((\frac{2}{2-\gamma})^{\frac{1}{k}}-1)$. This bound is
identical to the result $\ln(\frac{2}{2-\gamma})$ presented by Sha
\cite{DBLP:conf/rtcsa/Sha03} when $k \rightarrow \infty$.

If $T_{cycle} \geq T_k$, the virtual task $\tau_{virtual}$ created in
Sec. \ref{sec:segmented} should be part of $hp_2(\tau_k)$
defined in Sec.~\ref{sec:constant-inflation}. Therefore, the
schedulability test of task $\tau_k$ if $T_{cycle} \geq T_k$ is
\begin{align}
&  (U_k+\frac{T_{cycle}-C_{slot}}{T_k}+1)\prod_{\tau_i \in
    hp(\tau_k)}(U_i+1) \leq 2 \nonumber\\
\Rightarrow & \prod_{i=1}^{k-1} (U_i+1) \leq
  \frac{2}{1+U_k + \frac{T_{cycle}}{T_k}(1-\gamma)}.
  \label{eq:tdma-uni-rm-cycle-large}
\end{align}
If $T_{cycle} \geq T_k$ when $k$, we can conclude that task $\tau_k$
is schedulable under RM scheduling if $\sum_{i=1}^{k-1} U_i \leq \ln(2) -  \ln(1+U_k + \frac{T_{cycle}}{T_k}(1-\gamma))$.

For the case with the bounded delay service curve, we can use Eq.~\eqref{eq:test-bounded-delay} with $t_{delay}=T_{cycle}-C_{slot}$, $\sigma^*=1$,
$b=0$, and $\gamma=T_{cycle}/C_{slot}$. This results in the following
schedulability test by using Corollary~\ref{cor:inflation} for RM scheduling
\begin{equation}
  \label{eq:bounded-delay-uni-rm}
  \left(\frac{C_k + \gamma t_{delay}}{\gamma T_k} + 1\right)
  \prod_{i=1}^{k-1} \left(\frac{U_i}{\gamma}+1\right) \leq 2.
\end{equation}

Therefore, if $\frac{t_{delay}}{T_k}$ is negligible, i.e., the TDMA cycle
is extremely shorter than $T_k$, then, we can conclude a utilization
bound of $\gamma \ln 2$, which dominates
$\ln(\frac{2}{2-\gamma})$. However, if $t_{delay}$ is very close to $T_k$,
then the test in Eq.~\eqref{eq:tdma-uni-rm} is better.

Note that the above treatment can be easily extended to handle
deferrable servers, sporadic servers, polling servers, and
constrained-deadline task systems. Extending the analysis to
multiprocessor systems is also possible if the schedulability test
can be written as Eq.~\eqref{eq:test-bounded-delay}.

\subsection{Arrival Jitter}
\label{sec:jitter}
Suppose that the schedulability test is as follows:
\begin{equation}
  \label{eq:test-arrival-jitter}
  \exists 0 < t \leq D_k \mbox{ s.t. } C_k + \sum_{\tau_i \in hp(\tau_k)}\sigma \left(\ceiling{\frac{t+ \delta T_i}{T_i}} C_i \right) \leq t,
\end{equation}
where $\sigma > 0$ and $\delta \geq 0$. Note that if $\delta$ is an
integer, then this is a special case of
Eq.~\eqref{eq:test-constant-inflation}. We will first focus on the cases when
$\delta$ is not an integer.  We again classify the task set
$hp(\tau_k)$ into two subsets:
\begin{itemize}
\item $hp_2(\tau_k)$ consists of the higher-priority tasks $\tau_i$
  with $\ceiling{\frac{D_k + \delta T_i}{T_i}}$  equal to $\ceiling{\delta}$.
\item $hp_1(\tau_k)$ is $hp(\tau_k)\setminus hp_2(\tau_k)$.
\end{itemize}
Therefore, we can rewrite Eq.~\eqref{eq:test-arrival-jitter} to 
\begin{equation}
  \label{eq:test-arrival-jitter2}
  \exists 0 < t \leq D_k \mbox{ s.t. } C_k' + \sum_{\tau_i \in hp_1(\tau_k)}\sigma \left(\ceiling{\frac{t+\delta T_i}{T_i}} C_i\right) \leq t,
\end{equation}
where $C_k'$ is defined as $C_k + \sum_{\tau_i \in hp_2(\tau_k)}\sigma \ceiling{\delta}C_i$.

\begin{theorem}
  \label{thm:arrival-jitter}
  For
  Eq.~\eqref{eq:test-arrival-jitter2}, the $k$-point effective
  schedulability test in Definition~\ref{def:kpoints-k2u} is with the
  following settings:
  \begin{compactitem}
  \item $t_k = D_k$,
  \item  for $\tau_i \in hp_1(\tau_k)$, set $g_i =  \floor{\frac{D_k+\delta T_i}{T_i}}$,
  \item  for $\tau_i \in hp_1(\tau_k)$, set $t_i =\left
      (\floor{\frac{D_k+\delta T_i}{T_i}}-\delta\right)T_i = (g_i-\delta) T_i$,
  \item  for $\tau_i \in hp_1(\tau_k)$,  the parameter $\alpha_i$ is $\frac{\sigma g_i}{g_i-\delta}$ with $0 < \alpha_i \leq \frac{\sigma\ceiling{\delta}}{\ceiling{\delta}-\delta}$, and 
  \item  for $\tau_i \in hp_1(\tau_k)$,  the parameter $\beta_i$ is $\frac{\sigma}{g_i-\delta}$ with $0 < \beta_i \leq \frac{\sigma}{\ceiling{\delta}-\delta}$.
  \end{compactitem}
  The tasks in $hp_1(\tau_k)$ are indexed according to non-decreasing
  $t_i$ defined above to satisfy Definition~\ref{def:kpoints-k2u}.
\end{theorem}
\begin{proof}
  By the definition of $t_i$ and $hp_1(\tau_k)$, we know that $g_i$ is
  an integer with $g_i > \delta$. By following the same procedure in
  the proof of Theorem~\ref{thm:constant-inflation}, the left-hand side in
  Eq.~\eqref{eq:test-arrival-jitter2} at time $t=t_j$ is upper bounded by
  whether {\small \begin{align}
      &  C_k' + \sum_{i=1}^{k-1}\sigma \left(\ceiling{\frac{t_j + \delta T_i}{T_i}} C_i\right)\nonumber\\
      \leq \;\;& C_k' + \sum_{i=1}^{j-1}\sigma \left(\ceiling{\frac{D_k + \delta T_i}{T_i}} C_i \right) + \sum_{i=j}^{k-1}\sigma \left(\ceiling{\frac{t_i + \delta T_i}{T_i}} C_i\right)\nonumber\\
      \leq \;\;& C_k' + \sum_{i=1}^{j-1}\sigma (g_i+1) C_i + \sum_{i=j}^{k-1}\sigma g_i C_i\nonumber\\
      = \;\;& C_k' + \sum_{i=1}^{k-1}\sigma g_i C_i + \sum_{i=1}^{j-1} \sigma C_i \nonumber\\
      =_1 & C_k' + \sum_{i=1}^{k-1} \frac{\sigma g_i}{g_i-\delta} t_i
      U_i + \sum_{i=1}^{j-1} \frac{\sigma}{g_i-\delta}t_i U_i,
    \end{align}} where the last equality comes from the setting that
  $C_i = T_i U_i= \frac{1}{g_i-\delta} t_i U_i$.

  It is not difficult to see that $\frac{1}{g_i-\delta}$ and
  $\frac{g_i}{g_i-\delta}$ are both decreasing functions with respect
  to $g_i$ if $g_i > \delta$. Therefore, we know that $0 < \alpha_i
  \leq \frac{\sigma\ceiling{\delta}}{\ceiling{\delta}-\delta}$ and $0
  < \beta_i \leq \frac{\sigma}{\ceiling{\delta}-\delta}$ since $g_i$
  is an integer. We therefore conclude the proof.
\end{proof}

The above analysis may be improved by further annotating
$hp_1(\tau_k)$ to enforce $g_i > \delta+1$ if $\ceiling{\delta}$ is
very close to $\delta$. 

\begin{corollary}
  Suppose that we classify the task set $hp(\tau_k)$ into two subsets:
  \begin{compactitem}
  \item $hp_2(\tau_k)$ consists of the higher-priority tasks $\tau_i$
    with $\ceiling{\frac{D_k + \delta T_i}{T_i}}$ less than or equal to $\ceiling{\delta}+1$.
  \item $hp_1(\tau_k)$ is $hp(\tau_k)\setminus hp_2(\tau_k)$.
  \end{compactitem}
  Then, for each task $\tau_i \in hp_1(\tau_k)$, we have $0 < \alpha_i
  \leq \frac{\sigma(\ceiling{\delta}+1)}{\ceiling{\delta}+1-\delta}$
  and $0 < \beta_i \leq \frac{\sigma}{\ceiling{\delta}+1-\delta}$ for
  the schedulability test in Eq.~\eqref{eq:test-arrival-jitter2}, 
  where $C_k'$ is defined as $C_k + \sum_{\tau_i \in hp_2(\tau_k)}\sigma \ceiling{\frac{D_k+\delta T_i}{T_i}}C_i$.
\end{corollary}
\begin{proof}
  This is identical to the proof of Theorem~\ref{thm:arrival-jitter} by using the fact $g_i > \delta+1$ for a task $\tau_i$ in $hp_1(\tau_k)$ defined in this corollary.
\end{proof}

The quantification of the arrival jitter in
Eq.~\eqref{eq:test-arrival-jitter} assumes an upper bounded jitter
$\delta T_i$ for each task $\tau_i \in hp(\tau_k)$. In many cases, the
higher-priority tasks have independent jitter terms. Putting the
arrival jitter of task $\tau_i$ to $\delta T_i$ is sometimes over
pessimistic. For the rest of this section, suppose that the
schedulability test is as follows:
\begin{equation}
  \label{eq:test-independent-jitter}
  \exists 0 < t \leq D_k \mbox{ s.t. } C_k + \sum_{\tau_i \in hp(\tau_k)}\sigma \left(\ceiling{\frac{t+ J_i}{T_i}} C_i \right) \leq t,
\end{equation}
where $\sigma > 0$ and $J_i \geq 0$ for every $\tau_i \in hp(\tau_k)$. We again classify the task set
$hp(\tau_k)$ into two subsets:
\begin{itemize}
\item $hp_2(\tau_k)$ consists of the higher-priority tasks $\tau_i$
  with $\ceiling{\frac{D_k + J_i}{T_i}}$ equal to $\ceiling{J_i/T_i}$.
\item $hp_1(\tau_k)$ is $hp(\tau_k)\setminus hp_2(\tau_k)$.
\end{itemize}
Therefore, we can rewrite Eq.~\eqref{eq:test-independent-jitter} to 
\begin{equation}
  \label{eq:test-independent-jitter2}
  \exists 0 < t \leq D_k \mbox{ s.t. } C_k' + \sum_{\tau_i \in hp_1(\tau_k)}\sigma \left(\ceiling{\frac{t+J_i}{T_i}} C_i\right) \leq t,
\end{equation}
where $C_k'$ is defined as $C_k + \sum_{\tau_i \in hp_2(\tau_k)}\sigma \ceiling{J_i/T_i}C_i$

\begin{theorem}
  \label{thm:independent-jitter}
  For each task $\tau_i$ in $hp_1(\tau_k)$ in
  Eq.~\eqref{eq:test-independent-jitter2}, the $k$-point effective
  schedulability test in Definition~\ref{def:kpoints-k2u} is with the
  following settings:
  \begin{compactitem}
  \item $t_k = D_k$,
  \item  for $\tau_i \in hp_1(\tau_k)$, set $g_i =  \floor{\frac{D_k+J_i}{T_i}}$,
  \item  for $\tau_i \in hp_1(\tau_k)$, set $t_i =\floor{\frac{D_k+J_i}{T_i}}T_i - J_i = g_i T_i - J_i$,
  \item  for $\tau_i \in hp_1(\tau_k)$,  the parameter $\alpha_i$ is $\frac{\sigma g_i}{g_i-J_i/T_i}$, and
  \item  for $\tau_i \in hp_1(\tau_k)$,  the parameter $\beta_i$ is $\frac{\sigma}{g_i- J_i/T_i}$.
  \end{compactitem}
  The tasks in $hp_1(\tau_k)$ are indexed according to non-decreasing
  $t_i$ defined above to satisfy Definition~\ref{def:kpoints-k2u}.
\end{theorem}
\begin{proof}
  The proof is identical to that of Theorem~\ref{thm:arrival-jitter}.
\end{proof}

{\bf Applications:} Arrival jitter is very common in task systems,
especially when no critical instant theorem has been
established. Therefore, instead of exploring all the combinations of
the arrival times of the higher-priority tasks, quantifying the
scheduling penalty with a jitter term is a common approach. For example,
in a self-suspending constrained-deadline sporadic task system, we can
quantify the arrival jitter $J_i$ of the higher-priority task $\tau_i
\in hp(\tau_k)$ as $D_i-C_i$ by assuming that $\tau_i$ meets its
deadline, e.g., \cite{huangpass:dac2015}. Suppose that $S_i$ is the
self-suspension time of a task $\tau_i$. For a self-suspending
implicit-deadline task system under fixed-priority scheduling, it is
shown in \cite{huangpass:dac2015} that the schedulability test is to
verify
\begin{equation}
  \label{eq:test-suspending-jitter}
  \exists 0 < t \leq T_k \mbox{ s.t. } C_k+S_k + \sum_{\tau_i \in hp(\tau_k)} \left(\ceiling{\frac{t+ T_i-C_i}{T_i}} C_i \right) \leq t. 
\end{equation}
That is, $\sigma=1$ and $J_i$ is $T_i - C_i$ in Eq.~\eqref{eq:test-independent-jitter}. Therefore, we can use
Theorem~\ref{thm:independent-jitter} to construct a polynomial-time
schedulability test.

\section{Conclusion}
\label{sec:conclusion}

This report explains how to automatically derive the parameters needed
in the \frameworkku{} framework for several classes of widely used
schedulability tests. The procedure to derive the parameters was not
clear yet when we developed the \frameworkku{} framework in
\cite{DBLP:journals/corr/abs-1501.07084}.  Therefore, the parameters
in all the examples in \cite{DBLP:journals/corr/abs-1501.07084} were
manually constructed. This automation procedure
significantly empowers the \frameworkku{} framework to automatically
handle a wide range of classes of real-time execution platforms and
task models, including uniprocessor scheduling, multiprocessor
scheduling, self-suspending task systems, real-time tasks with arrival
jitter, services and virtualizations with bounded delays, etc.

Moreover, we would also like to emphasize that the constructions of
the coefficients in this report may not be the best choices. We do not
provide any optimality guarantee of the resulting constructions. In
fact, given an arbitrary schedulability test, there are many ways to
define a corresponding k-point effective schedulability test in
Definition~\ref{def:kpoints-k2u}. All the constructions in this report
follow the same design philosophy: \emph{We first identify the tasks
  that can release at least one more job at time $0 < t < D_k$ in the
  schedulability test and define the effective test point of such a
  task at its last release before $D_k$.} There may be other more
effective constructions for different schedulability tests. These
opportunities are not explored in this report.

\vspace{0.2in}
{\bf Acknowledgement}: This report has been supported by DFG, as
    part of the Collaborative Research Center SFB876
    (http://sfb876.tu-dortmund.de/), the priority program
    "Dependable Embedded Systems" (SPP 1500 -
    http://spp1500.itec.kit.edu), and NSF grants OISE 1427824 and CNS 1527727.  The authors would also like to
    thank to Mr. Niklas Ueter for his valuable feedback in the draft.

 \footnotesize
\bibliographystyle{abbrv} 
\bibliography{ref,real-time} 
\normalsize
\end{document}